\documentclass[a4paper,12pt]{article}
\usepackage{amsmath,amssymb,amsthm}
\usepackage[sort,authoryear]{natbib}

\makeatletter \oddsidemargin-.25in \evensidemargin-.25in
\makeatother

\usepackage{graphicx}

\newtheorem{proposition}{Proposition}
\newtheorem{theorem}{Theorem}
\begin{document}

\date{}
\title{\Large \textbf{{Pure cross-diffusion models:\\ Implications for traveling wave solutions}} }

\author{F. Berezovskaya$^{1,}$\footnote{e-mail: fberezovskaya@howard.edu}\,, G. Karev$^2$, and A. Novozhilov$^2$\\[2mm]
{\small \textit{$^1$Howard University, $6^{\mbox{th}}$ Str., Washington, DC 20059 USA}} \\
{\small $^2$\textit{National Institutes of Health, 8600 Rockville
Pike, Bethesda, MD 20894 USA}}} \maketitle

\begin{abstract}  An analysis of traveling wave solutions of pure
cross-diffusion systems, i.e., systems that lack reaction and
self-diffusion terms, is presented. Using the qualitative theory of
phase plane analysis the conditions for existence of different types
of wave solutions are formulated. In particular, it is shown that
family of wave trains is a generic phenomenon in pure
cross-diffusion systems. The results can be used for construction
and analysis of different mathematical models describing systems
with directional movement.
\end{abstract}

\section{Introduction}

Using so-called cross-diffusion terms in mathematical models in
natural sciences, in particular in biology, is nowadays a routine
trick to incorporate into classical models additional effects such
as taxis movement, crowding, pursuit, or evasion (for a review see
\cite{Tsyganov2007}). Two particular examples, which received
especial attention amongst other models, are the Keller--Segel model
of chemotactic movement
\cite{patlak1953rwp,Keller1971a,Horstmann2004} (\textit{taxis} is
defined as behavioral response to a directional stimulus), and the
Lotka--Volterra type cross-diffusion model to describe predator
pursuit of an evading prey
\cite{Kerner1959,Jorne1977,Dubey2001,horstmann2007rsl}.

A distinct feature of many studies on cross-diffusion models (PDE
systems where diffusion matrix is not diagonal) is that usually a
particular model is analyzed, where terms of reaction and diffusion
coefficients are given in fixed explicit form; notable exceptions
are given by
\cite{Berezovskaya1999b,Horstmann2004,BerezovskayaInPress}. It also
appears that almost all the models considered contain reaction
terms, only one exception we know of notwithstanding
\cite{Hadeler2007}, although it might be argued that systems without
reaction terms possess numerous solutions which may be relevant in
biological modeling.

Taking a step further, in this note we consider the systems that
lack the self-diffusion terms; PDE systems we analyze possess only
cross-diffusion (as, e.g., in \cite{Cartwright1997,Kuznetsov1994},
but with nonlinear cross-diffusion coefficients). We dub such
systems without reaction and self-diffusion terms \textit{pure
cross-diffusion systems}. In particular, we study the PDE system in
the following form:
\begin{equation}\label{cd1:1}
    \begin{split}
    u_t & =d_1(f(u,v)v_x)_x,\\
     v_t   &= d_2(g(u,v)u_x)_x,
\end{split}
\end{equation}
where $u=u(x,t),\,v=v(x,t),\,-\infty<x<\infty,\,t\geqslant 0,$ $d_1,
d_2$ are real constants, and $f(u,v),\,g(u,v)$ are arbitrary
functions whose properties will be specified later. Qualitatively
system \eqref{cd1:1} describes mutually dependent movement of
conserved entities $u$ and $v$ on gradients of their respective
counterparts. Of particular interest for us are the
\textit{traveling wave solutions} of system \eqref{cd1:1}, which,
e.g., can describe a replacement process of one species by another
species. As usual, a wave solution of \eqref{cd1:1} is a bounded
solution having the form
$u(x,t)=u(x+ct)=u(\xi),\,v(x,t)=v(x+ct)=v(\xi)$, where $c$ is the
wave speed. On substituting these traveling wave forms into
\eqref{cd1:1} we obtain the ODE wave system
\begin{equation}\label{cd1:2}
    \begin{split}
    cu_{\xi} & = d_1(f(u,v)v_{\xi})_{\xi},\\
    cv_{\xi} & =d_2(g(u,v)u_{\xi})_{\xi}.
\end{split}
\end{equation}
Note that we can integrate system \eqref{cd1:2}, and therefore,
\eqref{cd1:2} is essentially two dimensional, as the original PDE
system. Each wave solution of \eqref{cd1:1} has its counterpart as a
bounded orbit of \eqref{cd1:2}, which, due to the particular
structure of \eqref{cd1:2}, is an orbit on a phase plane. There is a
known correspondence between wave solutions of \eqref{cd1:1} and
bounded orbits of \eqref{cd1:2}: \textit{A wave front} in, e.g., $u$
component corresponds to a heteroclinic orbit of \eqref{cd1:2} that
connects singular points with different $u$ coordinates; \textit{A
wave impulse} in, e.g., $u$ component corresponds to a heteroclinic
orbit of \eqref{cd1:1} that connects singular points with identical
$u$ coordinates or to a homoclinic curve of a singular point;
\textit{A wave train} corresponds to a closed orbit of system
\eqref{cd1:2}.

In the next section we show that, under some natural restrictions on
functions $f(u,v)$ and $g(u,v)$, all possible traveling wave
solutions of \eqref{cd1:1} can be classified. In particular, a
family of wave trains of \eqref{cd1:1} is a generic phenomenon
intrinsic to \eqref{cd1:1}. Section 3 provides some examples and
extensions. Section 4 is devoted to discussion and conclusions.

\section{Wave solutions of the separable models}

In this note we confine ourselves to a particular case of
\eqref{cd1:1}, whereby the functions $f(u,v)$ and $g(u,v)$ can be
represented as products of functions that depend on one variable:
\begin{equation}\label{cd2:1}
f(u,v)=f_1(u)f_2(v),\,g(u,v)=g_1(u)g_2(v).
\end{equation}
Here the functions $f_i(x),\,g_i(x),\,i=1,2$ are smooth or rational.
The models of the form \eqref{cd1:1} for which \eqref{cd2:1} holds
will be termed \textit{separable pure cross diffusion models}, or
\textit{separable models} for short.

The rationale behind the choice of the from of the functions $f$ and
$g$ as in \eqref{cd2:1} is twofold. Firstly, such assumption allows
complete analysis of possible structures of the corresponding wave
system, and thus, exhaustive classification of possible wave
solutions of \eqref{cd1:1} as was done for different models in
\cite{BerezovskayaInPress,Berezovskaya2007}. Secondly, such forms
are abundant in the modeling literature being the consequence of a
biased random walk on the lattice (see, e.g., \cite{Othmer1997}).

For the separable model and after integration and rearrangement wave
system \eqref{cd1:2} takes the form
\begin{equation}\label{cd2:2}
    \begin{split}
    u_{\xi} & = cd_2(v-\eta)/g(u,v),\\
    v_{\xi} & = cd_1(u-\mu)/f(u,v),
\end{split}
\end{equation}
where $c\eta,\,c\mu$ are constants of integration that in general
depend on boundary conditions, and which are considered here as new
parameters of the system (see also \cite{BerezovskayaInPress}).
After the change of independent variable
$$
d\zeta=\frac{cd\xi}{d_1d_2f(u,v)g(u,v)}\,,
$$
which is well defined and smooth for any $(u,v)$ except where
$f(u,v)g(u,v)=0$ or where $f(u,v)$ or $g(u,v)$ are not defined (see
\cite{andronov1973qts}), we obtain
\begin{equation}\label{cd2:3}
    \begin{split}
    u_{\zeta} & = d_1f_1(u)f_2(v)(v-\eta),\\
    v_{\zeta} & = d_2g_1(u)g_2(v)(u-\mu).
\end{split}
\end{equation}
In the case if $f(u,v)$ or $g(u,v)$ are rational, the change of
independent variable is chosen such that the resulting wave system
is well defined for all $u$, $v$ (for examples see the next
section). Therefore, without loss of generality, we assume that our
wave system has the form
\begin{equation}\label{cd2:4}
    \begin{split}
    u_{\zeta} & = d_1\varphi_1(u)\varphi_2(v)(v-\eta),\\
    v_{\zeta} & = d_2\psi_1(u)\psi_2(v)(u-\mu).
\end{split}
\end{equation}
which is defined for all $u,\,v$.

First we note that system \eqref{cd2:4} is integrable. Function
\begin{equation*}
 H(u,v)=d_1\int(v-\eta)\frac{\varphi_2(v)}{\psi_2(v)}\,dv-{d_2}\int(u-\mu)\frac{\psi_1(u)}{\varphi_1(u)}\,du+\mbox{const}
\end{equation*}
represents its integral. We remark that under the assumption
$\varphi_1(u)=\mbox{const},\,\psi_2(v)=\mbox{const}$ system
\eqref{cd2:4} is hamiltonian, and $H(u,v)$ is its Hamiltonian, which
automatically yields that the singular points of \eqref{cd2:4} can
be only topological saddles or centers.

We assume that the following conditions of non-degeneracy are
fulfilled (these conditions are natural as far as we consider
$\eta,\,\mu,\,c$ the only parameters in the system \eqref{cd2:4}):
\begin{equation*}
    \begin{split}
(C1) &\quad \mbox{$\varphi_1(u)$ and $\psi_1(u)$ have no common roots;}\\
(C2) &\quad \mbox{$\varphi_2(v)$ and $\psi_2(v)$ have no common roots;}\\
(C3) &\quad \mbox{$\varphi_i(x)$ and $\psi_i(x),\, i=1,2,$ have no
multiple roots.}
\end{split}
\end{equation*}

System \eqref{cd2:4} always has singular point $(\mu,\,\eta)$.
Coordinates of other singular points are
$(\mu,\hat{v}_i),\,(\hat{u}_i,\,\eta),\,(\hat{u}_i,\hat{v}_i)$,
where $\hat{u}_i$ and $\hat{v}_i$ are the solutions of one of the
systems:
\begin{equation}\label{cd2:5}
    \varphi_1(u)=0,\quad \psi_2(v)=0,
\end{equation}
or
\begin{equation}\label{cd2:6}
    \varphi_2(v)=0,\quad \psi_1(u)=0.
\end{equation}

To infer the possible types of the singular points of \eqref{cd2:4}
consider the Jacobi matrix of \eqref{cd2:4}:
$DF(u,v)=(a_{ij})_{2\times 2},\,i,j=1,2$:
\begin{equation}\label{cd2:7}
    \begin{split}
a_{11}&=d_1(v-\eta)\varphi_2(v)(\varphi_1(u))'_u,\\
a_{12}&=d_1\varphi_1(u)((v-\eta)(\varphi_2(v))'_v+\varphi_2(v)),\\
a_{21}&=d_2\phi_2(v)((u-\mu)(\psi_1(u))'_u+\psi_1(u)),\\
a_{22}&=d_2(u-\mu)\psi_1(u)(\psi_2(v))'_v.
\end{split}
\end{equation}
Direct substitution of singular point coordinates into \eqref{cd2:7}
and evaluation of trace and determinant of $DF$ implies the
following proposition.

\begin{proposition}\label{pr2:1}Let us assume that conditions $(C1)$--$(C3)$
hold. We also assume that $\hat{u}_i$ \emph{(}$\hat{v}_i$\emph{)} do
not coincide with $\mu$ \emph{(}$\eta$, respectively\emph{)}. Then

1. Singular points of \eqref{cd2:4} satisfying \eqref{cd2:5} or
\eqref{cd2:6} can be saddles or nodes only;

2. Singular points of \eqref{cd2:4} where one or both coordinates
are $\eta$ or $\mu$ can be centers or saddles only.
\end{proposition}

\paragraph{Remark 1.} If we relax condition $(C3)$ retaining $(C1)$
and $(C2)$ then we have to add to the list of possible singular
points of \eqref{cd2:4} saddle-node points.

\paragraph{Remark 2.} To prove that singular points that have eigenvalues with
zero real part in linear approximation are indeed centers we can use
integrability of system \eqref{cd2:4} and Theorem 2 in \cite[p.
75]{Bautin1990}.\bigskip

From Proposition \ref{pr2:1}, integrability of system \eqref{cd2:4},
and special structure of the wave system \eqref{cd2:4} follows that
the structure of the phase plane, and, consequently, of the bounded
orbits of \eqref{cd2:4} are completely determined by the types and
mutual positions of singular points of \eqref{cd2:4}. The phase
plane is divided into orbit cells that represent either a
rectangular cell with orbit structure determined by the types of the
corner singular points (as in \cite{BerezovskayaInPress}), or a cell
filled with closed orbits surrounding singular point of the center
type (an example of such phase plane is given in Fig. \ref{f1:1}).
The former case corresponds to fronts or impulses of the initial
system \eqref{cd1:1}, and the latter to the wave train solutions of
system \eqref{cd1:1}.

\begin{figure}[tbh!]
\centering
\includegraphics[width=0.49\textwidth]{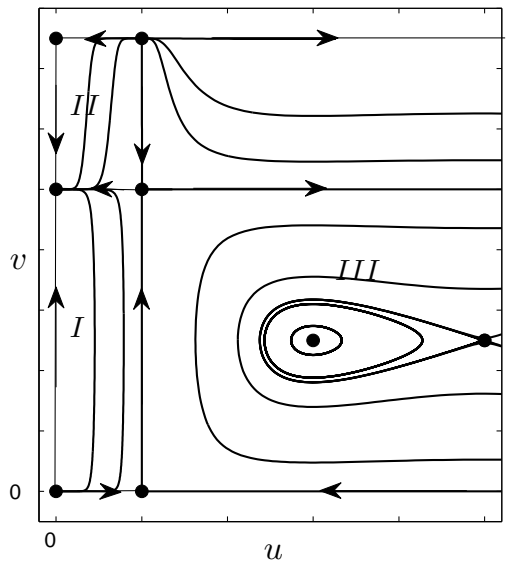}
\caption{A typical structure of the phase plane of the wave system
of the separable model \eqref{cd1:1}. Here the bold dots represent
singular points of \eqref{cd2:4}, the arrows show direction of the
vector field. The plane is divided into orbit cells that contain
families of bounded orbits: $I$ is a cell that corresponds to
impulse-front solutions; $II$ is a cell that corresponds to
front-front solutions; $III$ is a cell that corresponds to traveling
train solutions, the homoclinic surrounding the closed orbits
corresponds to impulse-impulse solutions}\label{f1:1}
\end{figure}

Summarizing, we obtain the basic result of the present note:
\begin{theorem}\label{th2:1}
The separable pure cross-diffusion system \eqref{cd1:1}, satisfying
$(C1)$--$(C3)$ can only possess traveling wave solutions of the
following kinds: i. impulse-front solutions; ii. front-front
solutions; iii. impulse-impulse solutions; iv. family of wave train
solutions.

The family of wave train solutions exits in domain $\Omega$ if
$f(u,v)$ and $g(u,v)$ are analytical in $\Omega$ and
$d_1d_2f(u,v)g(u,v)<0$ for $(u,v)\in\Omega$.
\end{theorem}

\section{Examples}

\paragraph{Example 1.} The simplest possible separable model \eqref{cd1:1}
is obtained if we put
\begin{equation}\label{cd3:1}
f_1(u)=f_2(v)=g_1(u)=g_2(v)=1.
\end{equation}
The integral in this case is given by
$$
d_2(u-\mu)^2-d_1(v-\eta)^2=\mbox{const}.
$$
From which it follows that system \eqref{cd1:1} with \eqref{cd3:1}
has no traveling wave solution if $d_1d_2>0$ and has a family of
traveling trains if $d_1d_2<0$.

\paragraph{Remark.} Note that in case $d_1d_2<0$ system \eqref{cd1:1}
is equivalent to the beam equation $u_{tt}=-Du_{xxxx},\,D>0$, which,
how it is well known, possesses wave solutions.

\paragraph{Example 2.} Let
\begin{equation}\label{cd3:2}
\begin{split}
f_2(v)&=g_1(u)=\mbox{const}\\
f_1(u)&=(u-a)^n,\,g_2(v)=(v-b)^m,\,n,m\in\mathbb N.
\end{split}
\end{equation}
It is straightforward to prove
\begin{proposition}\label{pr3:1}
Separable model \eqref{cd1:1} with \eqref{cd3:2} has a family of
wave trains if $d_1d_2<0$ and has a front-front solution, that
corresponds to a heteroclinic connecting node and saddle, if
$d_1d_2>0$.
\end{proposition}

\paragraph{Example 3.} Let
\begin{equation}\label{cd3:3}
\begin{split}
f_1(u)&=u^n,\,f_2(v)=\frac{1}{v^m}\,,\\
g_1(u)&=\frac{1}{u^n}\,,\,g_2(v)=v^m,
\end{split}
\end{equation}
where $m,n\in \mathbb N$. After a suitably chosen change of
independent variable in the wave system \eqref{cd2:2} we obtain
\begin{equation}\label{cd3:4}
\begin{split}
u_{\zeta}&=d_1u^{2n}(v-\eta),\\
v_{\zeta}&=d_2v^{2m}(u-\mu).
\end{split}
\end{equation}
Having in mind biological applications of the pure cross-diffusion
models we assume that $\mu>0,\,\eta>0$ and analyze the phase plane
in the first quadrant. System \eqref{cd3:4} has two singular points:
$(0,0)$ and $(\mu,\,\eta)$. The point $(\mu,\,\eta)$ can be a saddle
(if $d_1d_2>0$) or a center ($d_1d_2<0$). If we restrict our
attention only to the cases $\mu>0,\,\eta>0$, $u,v>0$, then
traveling wave solutions of \eqref{cd1:1} with \eqref{cd3:4} are
exactly the same as in Proposition \ref{pr3:1}.

\begin{figure}[tb!]
\centering
\includegraphics[width=0.46\textwidth]{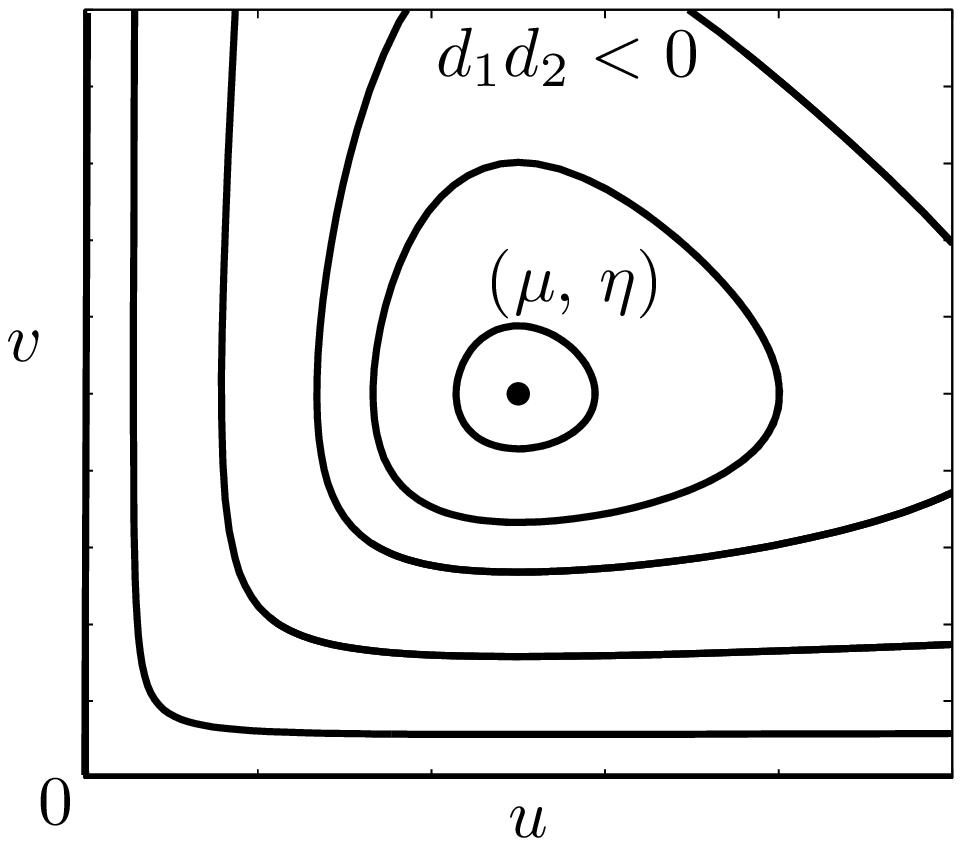}
\includegraphics[width=0.46\textwidth]{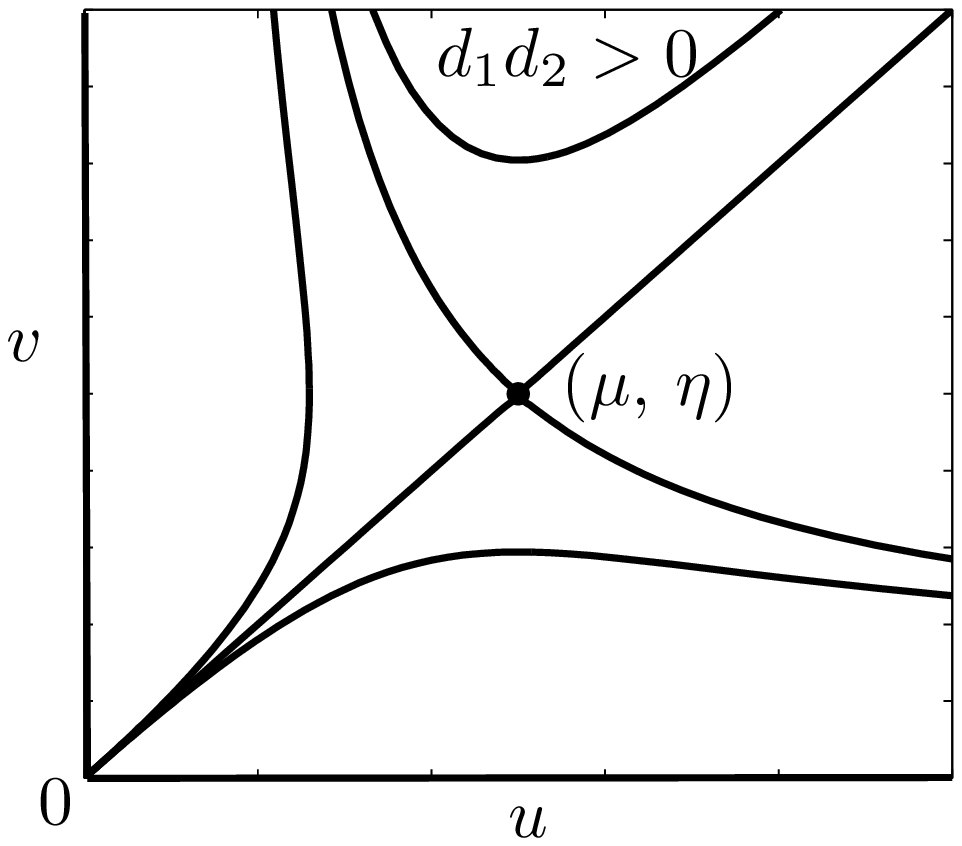}
\caption{The structure of the first quadrant of the phase plane of
the wave system \eqref{cd1:2}, \eqref{cd3:3}. In case $d_1d_2<0$ a
family of closed orbits can be found, which corresponds to the wave
trains; in case $d_1d_2>0$ there is only one bounded trajectory that
connect the origin and the saddle point $(\mu,\,\eta)$}\label{f3:1}
\end{figure}

Let us illustrate this case. In Fig. \ref{f3:1} the first quadrant
of the phase plane of the wave system of \eqref{cd1:1},\eqref{cd3:3}
is shown for two topologically nonequivalent cases. The first case
corresponds to the wave train in the original PDE system, and the
second case corresponds to front-front solution. In Fig. \ref{f3:2}
numerically obtained solutions of \eqref{cd1:1},\eqref{cd3:3} are
presented. To approximate cross-diffusion term an ``upwind''
explicit scheme was used \cite{morton2005nsp} which is often used
for cross-diffusion systems (e.g., \cite{Tsyganov2004a}).

\begin{figure}[b!]
\centering
\includegraphics[width=0.65\textwidth]{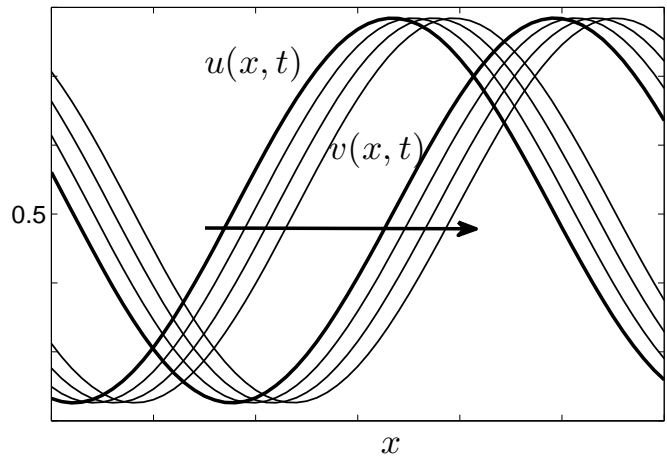}\\
\caption{Numerically obtained wave train solutions of system
\eqref{cd1:1},\eqref{cd3:3} at time moments $t_0=0$(bold
curves)$<t_1<t_2<t_3$ through equal time intervals. The arrow shows
the direction of the waves}\label{f3:2}
\end{figure}

\paragraph{Example 4.} We can use obtained results to prove
existence of solutions of a special form when simple reaction terms
are presented. For example, let us consider the following system:
\begin{equation}\label{cd3:5}
\begin{split}
u_t&=b_1u+d_1(u^n/v^mv_x)_x,\\
v_t&=b_2v+d_2(v^m/u^nu_x)_x.
\end{split}
\end{equation}
If the constant $b_1,\,b_2,\,n,\,m$ satisfy the following equality:
$$
b_1(1-n)=b_2(1-m),
$$
then by the change
$$
u(x,t)=e^{b_1t}z(z,t),\,v(x,t)=e^{b_2t}w(x,t)
$$
system \eqref{cd3:5} can be reduced to system \eqref{cd1:1},
\eqref{cd3:3}. Thus system \eqref{cd3:5} has wave train solutions
with amplitude changing with time. It is interesting to note that
for some models of the form \eqref{cd3:5} an exact solution can be
found.

\paragraph{Example 5.}Consider the system
\begin{equation}\label{cd3:6}
\begin{split}
u_t&=u+d_1v_{xx},\\
v_t&=-\gamma v+d_2u_{xx}.
\end{split}
\end{equation}
\begin{proposition}
For $\gamma<1$ and $d_1+d_2>0$ system \eqref{cd3:6} has a family of
wave trains.
\end{proposition}
\begin{proof}
We will prove proposition by presenting a solution of \eqref{cd3:6}
in explicit form.

Let $u(\xi)=\cos k\xi,\,v(\xi)=\sin k\xi,\,\xi=x+ct$. Then system
\eqref{cd3:6} transforms to the equations
\begin{equation*}
    \begin{split}
-ck\sin k\xi&=\sin k\xi-d_1k^2\sin k\xi,\\
ck\cos k\xi&=-\gamma \cos k\xi-d_2 k^2\cos k\xi,
\end{split}
\end{equation*}
from which
\begin{equation*}
    \begin{split}
d_1k^2-ck-1&=0,\\
d_2k^2+ck+\gamma&=0
\end{split}
\end{equation*}
follows. For fixed $d_1,\,d_2,\,\gamma$ the solution of the last
system is given by
$$
k=\frac{\sqrt{1-\gamma}}{d_1+d_2},\,c=-\frac{d_2+d_1\gamma}{\sqrt{(d_1+d_2)(1-\gamma)}},
$$
which is defined for $\gamma<1$ and $d_1+d_2>0$, which completes the
proof.
\end{proof}

\paragraph{Example 6.} An obvious question is what kind of solutions
can appear in pure cross-diffusion models if we do not demand that
the functions $f(u,v)$ and $g(u,v)$ can be represented as the
product of functions of one variable. Despite the fact that the
separable models possess very reach structure of possible traveling
wave solutions, simple extension of the models considered can be
shown to yield a different type of behavior.

Consider the model \eqref{cd1:1} with the functions given by
\begin{equation}\label{cd3:7}
\begin{split}
f(u,v)&=\frac{u}{\alpha u+v},\\
g(u,v)&=\frac{v}{u+\beta v},
\end{split}
\end{equation}
where $\alpha$ and $\beta$ are positive constant.

The wave system, after suitable change of independent variable,
reads
\begin{equation}\label{cd3:8}
\begin{split}
u_{\zeta}&=d_1(v-\eta)u(\alpha u+\eta),\\
v_{\zeta}&=d_2(u-\mu)v(u+\beta v).
\end{split}
\end{equation}

Assume that all the parameters, except for $d_1,\,d_2$ are positive.
Singular point $(\mu,\,\eta)$ is a saddle, if $d_1d_2>0$ and a
center at linear approximation for $d_1d_2<0$. However, contrast to
model in Example 3, and due to the fact that the wave system is now
non-integrable, we cannot make conclusion on the type of the
singular point in the case $d_1d_2<0$. If we calculate the first
Lyapunov value $l_1$ \cite{Bautin1990}, we obtain
$l_1=d_1d_2\mu\eta(\alpha\beta-1)(\alpha\mu^2-\beta\eta^2)$, which
shows that for parameter values where $l_1\neq0$ this singular point
can be either stable or unstable nonlinear focus.

\begin{figure}[t!]
\centering
\includegraphics[width=0.6\textwidth]{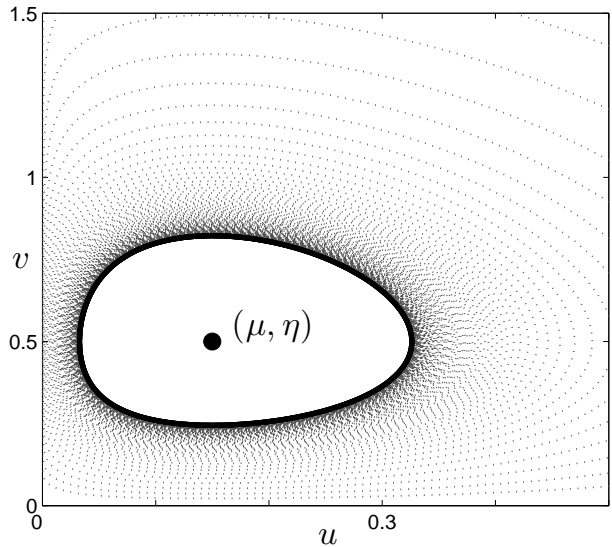}
\caption{Limit cycle in the wave system \eqref{cd3:8}. Parameter
values
$d_1=-1,\,d_2=1,\,\alpha=0.7,\,\beta=0.2,\,\mu=0.3,\,\eta=0.5$}\label{f3:3}
\end{figure}
Numerical experiments show that at least for some parameter values
there is a limit cycle surrounding the singular point (see Fig.
\ref{f3:3}). Thus, in this example, there is a wave train solution
of the initial pure cross-diffusion system.

\section{Conclusions}
In the present note we have shown that pure cross-diffusion models
possess very rich structure of possible traveling wave solutions
even in a rather simplified, though still realistic, case of
separable models, when the nonlinear cross-diffusion coefficients in
system \eqref{cd1:1} can be represented by product of functions
depending on one variable.

A distinct feature of pure cross diffusion models is the presence,
under some additional condition, of the family of wave train
solutions (such solutions were not found for several classed of
cross-diffusion models describing taxis, see
\cite{BerezovskayaInPress,Berezovskaya2007}). As numerical
experiments show, such wave trains can be observed at least for some
time interval. Together with families of wave trains, other typical
cases of traveling waves can be observed. In particular,
impulse-front, front-front, and impulse-impulse solutions are also a
possibility.

We emphasize that the separable model \eqref{cd1:1}, even when the
values of parameters $\mu,\eta,c$ are fixed, possesses, in general,
a family of traveling wave solutions; there are infinitely many
bounded orbits of \eqref{cd1:2} that correspond to traveling wave
solutions.

The presented results can be used for construction and analysis of
different mathematical models describing systems where directional
movement is of particular importance.

\end{document}